\documentclass[11pt]{article}

\usepackage{amssymb,amsmath,amsthm,sectsty,url}
\usepackage[letterpaper,hmargin=1.0in,vmargin=1.0in]{geometry}
\usepackage[pdftex,colorlinks,linkcolor=blue,citecolor=blue,filecolor=blue,urlcolor=blue]{hyperref}
\usepackage{cleveref}
\usepackage{color}
\usepackage[boxed]{algorithm}
\usepackage{tikz}
\usepackage{graphicx}
\usepackage{nicefrac}
\usepackage{aliascnt}

\newtheorem{theorem}{Theorem}[section]
\crefname{theorem}{Theorem}{Theorems}

\newaliascnt{lemma}{theorem}
\newtheorem{lemma}[lemma]{Lemma}
\aliascntresetthe{lemma}
\crefname{lemma}{Lemma}{Lemmas}

\newaliascnt{proposition}{theorem}
\newtheorem{proposition}[proposition]{Proposition}
\aliascntresetthe{proposition}
\crefname{proposition}{Proposition}{Propositions}

\newaliascnt{corollary}{theorem}
\newtheorem{corollary}[corollary]{Corollary}
\aliascntresetthe{corollary}
\crefname{corollary}{Corollary}{Corollaries}

\newaliascnt{fact}{theorem}

\aliascntresetthe{fact}
\crefname{fact}{Fact}{Facts}

\newaliascnt{definition}{theorem}
\newtheorem{definition}[definition]{Definition}
\aliascntresetthe{definition}
\crefname{definition}{Definition}{Definitions}

\newaliascnt{remark}{theorem}

\aliascntresetthe{remark}
\crefname{remark}{Remark}{Remarks}

\newaliascnt{conjecture}{theorem}

\aliascntresetthe{conjecture}
\crefname{conjecture}{Conjecture}{Conjectures}

\newaliascnt{claim}{theorem}
\newtheorem{claim}[claim]{Claim}
\aliascntresetthe{claim}
\crefname{claim}{Claim}{Claims}

\newaliascnt{question}{theorem}

\aliascntresetthe{question}
\crefname{question}{Question}{Questions}

\newaliascnt{exercise}{theorem}

\aliascntresetthe{exercise}
\crefname{exercise}{Exercise}{Exercises}

\newaliascnt{example}{theorem}

\aliascntresetthe{example}
\crefname{example}{Example}{Examples}

\newaliascnt{notation}{theorem}

\aliascntresetthe{notation}
\crefname{notation}{Notation}{Notations}

\newaliascnt{problem}{theorem}

\aliascntresetthe{problem}
\crefname{problem}{Problem}{Problems}

\newcommand{\norm}[1]{\lVert#1\rVert}

\newcommand{\R}{\mathbb R}

\newcommand{\B}{\{ 0,1 \}}

\title{Approximate Nearest Neighbors in Limited  Space}
\author{
  Piotr Indyk\thanks{\texttt{indyk@mit.edu}} \\
  MIT \\
  \and
  Tal Wagner\thanks{\texttt{talw@mit.edu}} \\
  MIT \\
}

\begin{document}
\maketitle

\begin{abstract}
We consider the $(1+\epsilon)$-approximate nearest neighbor search problem: given a set $X$ of $n$ points in a $d$-dimensional space, build a data structure that, given any query point  $y$,  finds a point $x \in X$ whose distance to $y$ is at most $(1+\epsilon) \min_{x \in X} \|x-y\|$ for an accuracy parameter $\epsilon \in (0,1)$.  Our main result is a data structure that occupies only $O(\epsilon^{-2} n \log(n) \log(1/\epsilon))$
bits of space, assuming all point coordinates are integers in the range  $\{-n^{O(1)} \ldots n^{O(1)}\}$, i.e., the coordinates have $O(\log n)$ bits of precision. This improves over the best previously known space bound of         $O(\epsilon^{-2} n \log(n)^2)$, obtained via the randomized dimensionality reduction method of~\cite{johnson1984extensions}.  We also consider the more general problem of estimating all distances from a collection of query points to all data points $X$, and provide almost tight upper and lower bounds for the space complexity of this problem. 
\end{abstract}

\section{Introduction}

The nearest neighbor search problem is defined as follows: given a set $X$ of $n$ points in a $d$-dimensional space, build a data structure that, given any query point $y$, returns the point in $X$ closest to $y$. For efficiency reasons, the problem is often relaxed to approximate nearest neighbor search, where the goal is to find a point $x \in X$ whose distance to $y$ is at most $c \min_{x \in X} \|x-y\|$ for some approximation factor $c>1$. Both problems have found numerous applications in machine learning, computer vision,  information retrieval and other areas. In machine learning in particular, nearest neighbor classifiers are popular baseline methods whose classification error often comes close to that of the best known techniques (\cite{efros}).


Developing fast approximate nearest neighbor search algorithms have been a subject of extensive research efforts over the last last two decades, see  e.g.,~\cite{shakhnarovich2006nearest,AI} for an overview.  More recently, there has been increased focus on  designing nearest neighbor methods that use a limited amount of {\em space}. This is motivated by the need to fit  the data set in the main memory (\cite{johnson2017billion,faiss}) or an Internet of Things device (\cite{gupta2017protonn}).  Furthermore, even a simple linear scan over the data is more time-efficient if the data is compressed. 
The data set compression is most often achieved 
by developing compact representations of data that approximately preserve the distances between the points (see~\cite{wang2016learning} for a survey). 
Such representations are smaller than the original (uncompressed) representation of the data set, while  approximately preserving the distances between points.  

Most of the approaches in the literature are only validated empirically.
The currently best known {\em theoretical} tradeoffs between the representation size and the approximation quality are summarized in Table~\ref{tbl:sketches_related_work}, together with their functionalities and constraints. 
\begin{table}[h]
\begin{center}
\begin{tabular}{| l | l | l |}
\hline
& 		Bits per point	& Comments \\
\hline
No compression 		& $d  \log n$ 							&  \\
\cite{johnson1984extensions}	& $\epsilon^{-2} \log^2(n)$			& Estimates distances between any $y$ and all $x \in X$\\
 \cite{kushilevitz2000efficient}				& $\epsilon^{-2} \log(n) \log (R)$				& Estimates distances between any $y$ and all $x \in X$,  \\
 & & assuming $\|x-y\| \in [r, R r]$ \\
 \cite{indyk2017near}				& $\epsilon^{-2} \log(n) \log(1/\epsilon)$	& Estimates distances between all $x,y \in X$,  \\
 & & does not provably support out-of-sample queries \\
 \hline
 This paper			& $\epsilon^{-2} \log(n) \log(1/\epsilon)$      & Returns an approximate nearest neighbor of $y$ in $X$\\
 \hline
\end{tabular}
\caption{Comparison of Euclidean metric sketches with distortion $1\pm\epsilon$. We assume that all point coordinates are represented using $\log \Phi$ bits, or alternatively that each coordinate is an integer in the range $\{ -\Phi \ldots \Phi\}$.
For the sake of exposition, the results depicted in the table assume $\Phi=n^{O(1)}$. Furthermore, the compression algorithm can be randomized, and the compressed representation must enable approximating distances up to a factor of $1 \pm \epsilon$ with probability $1/n^{O(1)}$. 
}
\label{tbl:sketches_related_work}
\end{center}
\end{table}

Unfortunately,  in the context of approximate nearest neighbor search, the above representations lead to sub-optimal results. 
The result from the last row of the table (from \cite{indyk2017near}) cannot be used to obtain provable bounds for nearest neighbor search, because the distance preservation guarantees hold only for pairs of points in the pointset $X$.\footnote{We note, however, that a simplified version of this method, described in~\cite{indyk2017practical}, was shown to have  good empirical performance for nearest neighbor search.} 
The second-to-last result (from \cite{kushilevitz2000efficient}) only estimates distances in a certain range; extending this approach to all distances would multiply the storage by a factor of $\log \Phi$.  Finally,  the representations obtained via a direct application of  randomized dimensionality reduction (\cite{johnson1984extensions}) are also larger than the bound from~\cite{indyk2017near} by almost a factor of $\log \Phi$. 

\paragraph{Our results} In this paper we show that it is possible to overcome the limitations of the previous results and design a compact representation that supports $(1+\epsilon)$-approximate nearest neighbor search, with a space bound essentially matching that of~\cite{indyk2017near}. This constitutes the first reduction in the space complexity of approximate nearest neighbor below the  ``Johnson-Lindenstrauss bound''.
Specifically, we show the following. Suppose that we want the data structure to answer $q$ approximate nearest neighbor queries in a $d$-dimensional dataset of size $n$, in which coordinates are represented by $\log\Phi$ bits each.
All $q$ queries must be answered correctly with probability $1-\delta$. (See Section~\ref{s:formal} for the formal problem definition).

\begin{theorem}\label{thm:ann_ub}
For the all-nearest-neighbors problem, there is a sketch of size 
\[
  O\left( n\left(\frac{\log n\cdot\log(1/\epsilon)}{\epsilon^2} + \log\log\Phi + \log\left(\frac{q}{\delta}\right)\right) + d\log\Phi  + \log\left(\frac{q}{\delta}\right)\log\left(\frac{\log(q/\delta)}{\epsilon}\right) \right)
  \;\; \text{bits.}
\]
\end{theorem}
The proof is given in~\Cref{sec:ann}. We also give a lower bound of
 $\Omega(n \log(n) /\epsilon^2)$ for $q=1$ and $\delta=1/n^{O(1)}$ (Section~\ref{sec:nnlower}), which shows that the first term in the above theorem is almost tight. 
 
Interestingly, the  representation by itself does not return the  (approximate) {\em distance} between the query point and the returned neighbor. 
Thus, we also consider the problem of estimating distances from a query point to all data points.
In this setting, a result of~\cite{molinaro2013beating} shows that the Johnson-Lindenstrauss space bound is optimal when the number of queries is~\emph{equal} to the number of data points.
However, in many settings, the number of queries is often substantially smaller than the dataset size.
We give nearly tight upper and lower bounds (up to a factor of $\log(1/\epsilon)$) for this problem, showing it is possible to smoothly interpolate between~\cite{indyk2017near}, which does not support out-of-sample distance queries, and the Johnson-Lindenstrauss bound. 

Specifically, we show the following. Suppose that we want the data structure to estimate all cross-distances between a set of $q$ queries and all points in $X$, all of which must be estimated correctly with probability $1-\delta$ (see Section~\ref{s:formal} for the formal problem definition).

\begin{theorem}\label{thm:distances_ub}
For the all-cross-distances problem, there is a sketch of size
\[ O\left(\frac{n}{\epsilon^2}\left(\log n\cdot\log(1/\epsilon) + \log(d\Phi)\log\left(\frac{q}{\delta}\right)\right) + \mathrm{poly}(d,\log\Phi,\log(q/\delta),\log(1/\epsilon))\right) \;\;  \text{bits.} \]
\end{theorem}
Note that  the dependence per point on $\Phi$ is logarithmic, as opposed to doubly logarithmic in Theorem~\ref{thm:ann_ub}.
We show this dependence is necessary, as per the following theorem. 
\begin{theorem}\label{thm:distances_lb}
Suppose that $d^{1-\rho}\geq\epsilon^{-2}\log(nq/\delta)$, $\Phi\geq1/\epsilon$, and $1/n^{0.5-\rho'}\leq\epsilon\leq\epsilon_0$ for some constants $\rho,\rho'>0$ and a sufficiently small constant $\epsilon_0$.
Then, for the all-cross-distances problem, any sketch must use at least
\[ \Omega\left(\frac{n}{\epsilon^2}\left(\log n + \log(d\Phi)\log\left(\frac{q}{\delta}\right)\right)\right) \;\; \text{bits.} \]
\end{theorem}
The proofs are given in ~\Cref{sec:dist} and \Cref{sec:dist_lb}, respectively. 

\paragraph{Practical variant}
\cite{indyk2017practical} presented a simplified version of~\cite{indyk2017near}, which has slightly weaker size guarantees, but on the other hand is practical to implement and was shown to work well empirically.
However, it did not provably support out-of-sample queries.
Our techniques in this paper can be adapted to their algorithm and endow it with such provable guarantees, while retaining its simplicity and practicality. We elaborate on this in Appendix~\ref{sec:middleout}.

\paragraph{Our techniques}
The starting point of our representation is the compressed tree data structure from~\cite{indyk2017near}.  The structure is obtained by constructing a hierarchical clustering of the data set, forming a tree of clusters.  The position of each point corresponding to a node in the tree is then represented by storing a (quantized) displacement vector between the point and its ``ancestor'' in the tree.
The resulting tree is further compressed by identifying and post-processing ``long'' paths in the tree. 
The intuition is that a subtree at the bottom of such a path corresponds to a cluster of points that is ``sufficiently separated'' from the rest of the points (see Figure~\ref{fig:noquery}). This means that the data structure does not need to know the exact position of this cluster in order to estimate the distances between the points in the cluster and the rest of the data set. Thus the data structure replaces each long path by a quantized displacement vector, where the quantization error does not depend on the length of the path. This ensures that the tree does not have long paths, which bounds its total size. 

Unfortunately, this reasoning breaks down if one of the points  is not known in advance, as it is the case for the approximate nearest neighbor problem. In particular, if the query point $y$ lies in the vicinity of the separated cluster, then small perturbations to the cluster position can dramatically affect which points in the cluster are closest to $y$ (see Figure~\ref{fig:query} for an illustration). 

In this paper we overcome this issue by maintaining extra information about the geometry of the point set. First, for each long path, we store not only the quantized displacement vector (which preserves the ``global'' position of the subtree with respect to the rest of the tree) but also the {\em suffix} of the path. Intuitively, this allows us to recover both the {\em most} significant bits and the {\em least} significant bits of  points in the subtree corresponding to the ``separated" clusters, which allows us to avoid cases as depicted in Figure~\ref{fig:query}. However, this intuition breaks down when the diameter of the cluster is much larger than the amount of ``separation".  Thus we also need to store extra information about the position of the subtree points. This is accomplished by storing a hashed representation of a representative point of the subtree (called ``the center"). We note that this modification makes our data structure inherently randomized; in contrast, the data structure of~\cite{indyk2017near} was deterministic.

Given the above information, the approximate nearest neighbor search is performed top down, as follows. In each step, we recover and enumerate points  in the current subtree, some of which could be centers of ``separated'' clusters as described above. The ``correct'' center, guaranteed to contain an approximate nearest neighbor of the query point, is identified by its hashed value (if no hash match is found, then any center is equally good). Note that our data structure does not allow us to compute all distances from the query point $y$ to all points in $X$ (in fact, as mentioned earlier, this task is not possible to achieve within the desired space bound). Instead, it stores just enough information to ensure that the procedure never selects a ``wrong" subtree to iterate on.

Lastly, suppose we also wish to estimate all distances from $y$ to $X$.
To this end, we augment each subtree with the distance sketches due to~\cite{kushilevitz2000efficient} and~\cite{johnson1984extensions}.
The former allows us to identify the cluster of~\emph{all} approximate nearest neighbors of $y$ (whereas the above algorithm was only guaranteed to return~\emph{one} approximate nearest neighbor).
The latter stores the approximate distance from that cluster.
These are the smallest distances from $y$ to $X$, which are the most challenging to estimate; the remaining distances can be estimated based on the hierarchical partition into well-separated clusters, which is already present in the sketch.

\begin{figure}[t]
\vskip 0.2in
\begin{center}
\centerline{\includegraphics[scale=0.5]{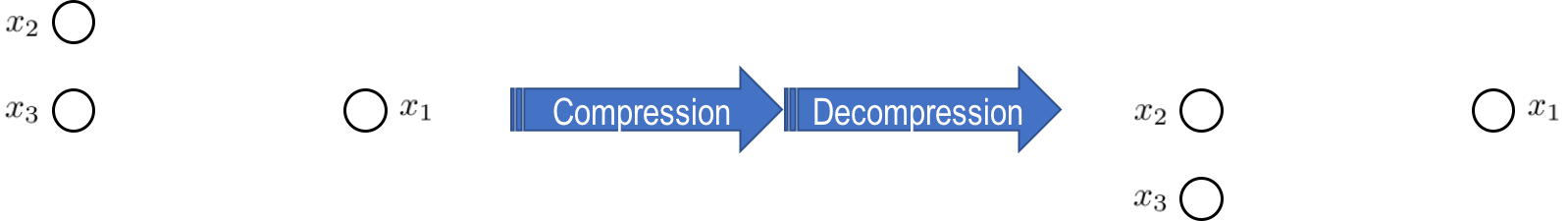}}
\caption{Compression and decompression of a two-dimensional dataset. The location of the well-separated cluster $\{x_2,x_3\}$ can be perturbed by the lossy compression algorithm, without significantly changing the distances to $x_1$.}
\label{fig:noquery}
\end{center}
\vskip -0.2in
\end{figure} 

\begin{figure}
\vskip 0.2in
\begin{center}
\centerline{\includegraphics[scale=0.5]{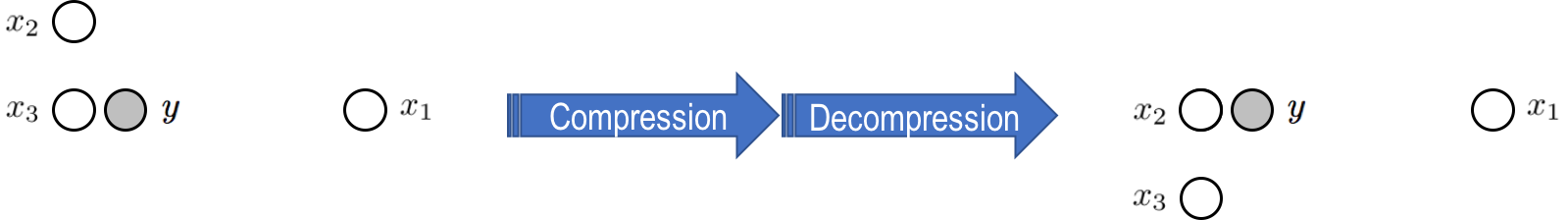}}
\caption{Compression and decompression of a dataset $x_1,x_2,x_3$ in the presence of a new query point $y$, which is unknown during compression. The same small perturbation in the location of $\{x_2,x_3\}$ as in Figure~\ref{fig:noquery} fails to preserve $x_3$ as the nearest neighbor of $y$.}
\label{fig:query}
\end{center}
\vskip -0.2in
\end{figure} 

\section{Formal Problem Statements}
\label{s:formal}
We formalize the problems in terms of one-way communication complexity.
The setting is as follows.
Alice has $n$ data points, $X=\{x_1,\ldots,x_n\} \subset \{-\Phi \ldots \Phi\}^d$, while
Bob has $q$ query points, $Y=\{y_1,\ldots,y_q\}  \subset \{-\Phi \ldots \Phi\}^d$, where $1\leq q\leq n$.
Distances are Euclidean, and we can assume w.l.o.g.~that $d\leq n$.\footnote{Any $N$-point Euclidean metric can be embedded into $N-1$ dimensions.}
Let $\epsilon,\delta\in(0,1)$ be given parameters.
In the one-way communication model, Alice computes a compact representation (called a~\emph{sketch}) of her data points and sends it to Bob, who then needs to report the output.
We define two problems in this model (with private randomness), each parameterized by $n,q,d,\Phi,\epsilon,\delta$:\footnote{Throughout we use $[m]$ to denote $\{1,\ldots,m\}$, for an integer $m>0$.}

\paragraph{Problem 1 -- All-nearest-neighbors:}
Bob needs to report a $(1+\epsilon)$-approximate nearest neighbor in $X$ for all his points simultaneously, with probability $1-\delta$.
That is, for every $j\in[q]$, Bob reports an index $i_j\in[n]$ such that 
\[ \Pr\left[ \forall j\in[q], \;\; \norm{y_j-x_{i_j}}\leq(1+\epsilon)\min_{i\in[n]}\norm{y_j-x_i} \right] \geq 1-\delta . \]

Our upper bound for this problem is stated in~\Cref{thm:ann_ub}.

\paragraph{Problem 2 -- All-cross-distancess:}
Bob needs to estimate all distances $\norm{x_i-y_j}$ up to distortion $(1\pm\epsilon)$ simultaneously, with probability $1-\delta$.
That is, for every $i\in[n]$ and $j\in[q]$, Bob reports an estimate $E(i,j)$ such that
\[ \Pr\Big[ \forall i\in[n],j\in[q], \;\; (1-\epsilon)\norm{x_i-y_j} \leq E(i,j) \leq (1+\epsilon)\norm{x_i-y_j} \Big] \geq 1-\delta . \]
Our upper and lower bounds for this problem are stated in~\Cref{thm:distances_ub,thm:distances_lb}.

\section{Basic Sketch}\label{sec:sketch}

In this section we describe the basic data structure (generated by Alice) used for all of our results.
The data structure augments the representation from~\cite{indyk2017near}, which we will now reproduce.
For the sake of readability, the notions from the latter paper (tree construction via hierarchical clustering, centers, ingresses and surrogates) are interleaved with the new ideas introduced in this paper (top-out compression, grid quantization and surrogate hashing).
Proofs in this section are deferred to Appendix~\ref{sec:sketch_proofs}.


\subsection{Hierarchical Clustering Tree}
The sketch consists of an annotated hierarchical clustering tree, which we now describe with our modified ``top-out compression'' step.

\paragraph{Tree construction}
We construct the inter-link hierarchical clustering tree of $X$: In the bottom level (numbered $0$) every point is a singleton cluster, and level $\ell>0$ is formed from level $\ell-1$ by recursively merging any two clusters whose distance is at most $2^\ell$, until no two such clusters are present.
We repeat this until level $\lceil\log(2\sqrt d\Phi)\rceil$, even if all points in $X$ are already joined in one cluster at a lower level.
The following observation is immediate.
\begin{lemma}\label{lmm:separation}
If $x,x'\in X$ are in different clusters at level $\ell$, then $\norm{x-x'}\geq2^\ell$.
\end{lemma}

\paragraph{Notation}
Let $T^*$ denote the tree.
For every tree node $v$, we denote its level by $\ell(v)$, its associated cluster by $C(v)\subset X$, and its cluster diameter by $\Delta(v)$. For a point $x_i\in X$, let $\mathrm{leaf}(x_i)$ denote the tree leaf whose associated cluster is $\{x_i\}$.

\paragraph{Top-out compression}
The~\emph{degree} of a node in $T^*$ is its number of children.
A~\emph{$1$-path with $k$ edges} in  $T^*$ is a downward path $u_0,u_1,\ldots,u_k$, such that (i) each of the nodes $u_0,\ldots,u_{k-1}$ has degree $1$, (ii) $u_k$ has degree either $0$ or more than $1$, (iii) if $u_0$ is not the root of $T^*$, then its ancestor has degree more than $1$.

For every node $v$ denote $\Lambda(v) := \log(\Delta(v)/(2^{\ell(v)}\epsilon))$. 
If $v$ is the bottom of a $1$-path with more than $\Lambda(v)$ edges, we replace all but the bottom $\Lambda(v)$ edges with a~\emph{long edge}, and annotate it by the length of the path it represents.
More precisely, if the downward $1$-path is $u_0,\ldots,u_k=v$ and $k>\Lambda(v)$, then we connect $u_0$ directly to $u_{k-\Lambda(v)}$ by the long edge, and the nodes $u_1,\ldots,u_{k-\Lambda(v)-1}$ are removed from the tree, and the long edge is annotated with length $k-\Lambda(v)$.

\begin{lemma}\label{lmm:tree_size}
The compressed tree has $O(n\log(1/\epsilon))$ nodes.
\end{lemma}

We henceforth refer only to the compressed tree, and denote it by $T$.
However, for every node $v$ in $T$, $\ell(v)$ continues to denote its level before compression (i.e., the level where the long edges are counted according to their lengths).
We partition $T$ into~\emph{subtrees} by removing the long edges.
Let $\mathcal F(T)$ denote the set of subtrees.

\begin{lemma}\label{lmm:subtree_root}
Let $v$ be the bottom node of a long edge, and $x,x'\in C(v)$. Then $\norm{x-x'}\leq2^{\ell(v)}\epsilon$.
\end{lemma}

\begin{lemma}\label{lmm:subtree_leaf}
Let $u$ be a leaf of a subtree in $\mathcal F(T)$, and $x,x'\in C(u)$. Then $\norm{x-x'}\leq2^{\ell(u)}\epsilon$.
\end{lemma}

\subsection{Surrogates}
The purpose of annotating the tree is to be able to recover a list of~\emph{surrogates} for every point in $X$.
A surrogate is a point whose location approximates $x$.
Since we will need to compare $x$ to a new query point, which is unknown during sketching, we define the surrogates to encompass a certain amount information about the absolute point location, by hashing a coarsened grid quantization of a representative point in each subtree.

\paragraph{Centers}
With every tree node $v$ we associate an index $c(v)\in[n]$ such that $x_{c(v)}\in C(v)$, and we call $x_{c(v)}$ the~\emph{center} of $C(v)$.
The centers are chosen bottom-up in $T$ as follows.
For a leaf $v$, $C(v)$ contains a single point $x_i\in X$, and we set $c(v)=i$. 
For a non-leaf $v$ with children $u_1,\ldots,u_k$, we set $c(v)=\min\{c(u_i):i\in[k]\}$.

\paragraph{Ingresses}
Fix a subtree $T'\in\mathcal F(T)$.
To every node $u$ in $T'$, except the root, we will now assign an~\emph{ingress} node, denoted $\mathrm{in(u)}$.
Intuitively this is a node in the same subtree whose center is close to $u$, and the purpose is to store the location of $u$ by its quantized displacement from that center (whose location will have been already stored, by induction).

We will now assign ingresses to all children of a given node $v$. (Doing this for every $v$ in $T'$ defines ingresses for all nodes in $T'$ except its root.) Let $u_1,\ldots,u_k$ be the children of $v$, and w.l.o.g.~$c(v)=c(u_1)$. Consider the graph $H_v$ whose nodes are $u_1,\ldots,u_k$, and $u_i,u_j$ are neighbors if there are points $x\in C(u_i)$ and $x'\in C(u_j)$ such that $\norm{x-x'}\leq2^{\ell(v)}$. By the tree construction, $H_v$ is connected. We fix an arbitrary spanning tree $\tau(v)$ of $H_v$ which is rooted at $u_1$.

For $u_1$ we set $\mathrm{in}(u_1):=v$. For $u_i$ with $i>1$, let $u_j$ be its (unique) direct ancestor in the tree $\tau(v)$. Let $x\in C(u_j)$ be the closest point to $C(u_i)$ in $C(u_j)$. Note that in $T$ there is a downward path from $u_j$ to $\mathrm{leaf}(x)$. Let $u_x$ be the bottom node in that path that belongs to $T'$. (Equivalently, $u_x$ is the bottom node on that downward path that is reachable from $u$ without traversing a long edge.) We set $\mathrm{in}(u_i):=u_x$.

\paragraph{Grid net quantization}
Assume w.l.o.g.~that $\Phi$ is a power of $2$.
We define a hierarchy of grids aligned with $\{-\Phi \ldots \Phi\}^d$ as follows.
We begin with the single hypercube whose corners are $(\pm\Phi, \ldots, \pm\Phi)^d$.
We generate the next grid by halving along each dimension, and so on.
For every $\gamma>0$, let $\mathcal{N}_\gamma$ be the coarsest grid generated, whose cell side is at most $\gamma/\sqrt{d}$. Note that every cell in $\mathcal{N}_\gamma$ has diameter at most $\gamma$.
For a point $x\in\R^d$, we denote by $\mathcal{N}_\gamma[x]$ the closest corner of the grid cell containing it.

We will rely on the following fact about the intersection size of a grid and a ball; see, for example, \cite{har2012approximate}.

\begin{claim}\label{clm:gridball}
For every $\gamma>0$, the number of points  in $\mathcal N_\gamma$ at distance at most $2\gamma$ from any given point, is at most $O(1)^d$.
\end{claim}

\paragraph{Surrogates}
Fix a subtree $T'\in\mathcal F(T)$. With every node $v$ in $T'$ we will now associate a~\emph{surrogate} $s^*(v)\in\R^d$.
Define the following for every node $v$ in $T'$:
\[
  \gamma(v) = 
  \begin{cases}
  \left(5 + \lceil\frac{\Delta(v)}{2^{\ell(v)}}\rceil\right)^{-1}\cdot\epsilon & \text{if $v$ is a leaf in $T'$,}\\
  \left(5 + \lceil\frac{\Delta(v)}{2^{\ell(v)}}\rceil\right)^{-1} & \text{otherwise.}
  \end{cases}
\]
The surrogates are defined by induction on the ingresses.

Induction base: For the root $v$ of $T'$ we set $s^*(v) := \mathcal N_{2^{\ell(v)}}[x_{c(v)}]$.

Induction step: For a non-root $v$ we denote the quantized displacement of $c(v)$ from its ingress by $\eta(v)=\mathcal N_{\gamma(v)}\left[\frac{\gamma(v)}{2^{\ell(v)}}(x_{c(v)}-s^*(in(v)))\right]$, and set
$s^*(v) := s^*(in(v)) + \frac{2^{\ell(v)}}{\gamma(v)}\cdot\eta(v)$.

\begin{lemma}\label{lmm:surrogates}
For every node $v$, $\norm{x_{c(v)}-s^*(v)}\leq2^{\ell(v)}$.
Furthermore if $v$ is a leaf of a subtree in $\mathcal F(T)$, then $\norm{x_{c(v)}-s^*(v)}\leq2^{\ell(v)}\epsilon$.
\end{lemma}

\paragraph{Hash functions}
For every level $\ell$ in the tree, we pick a hash function $H_\ell:\mathcal N_{2^\ell}\rightarrow[m]$, from a universal family (\cite{carter1979universal}), where $m=O(1)^d\cdot\log(2\sqrt d\Phi)\cdot q/\delta$.
The $O(1)$ term is the same constant from~\Cref{clm:gridball} above.
For every subtree root $v$, we store its hashed surrogate $H_{\ell(v)}(\mathcal N_{2^{\ell(v)}}[x_{c(v)}])$.
We also store the description of each hash function $H_\ell$ for every level $\ell$.

\subsection{Sketch Size}
The sketch contains the tree $T$, with each node $v$ annotated by its center $c(v)$, ingress $\mathrm{in(u)}$, precision $\gamma(v)$ and quantized displacement $\eta(v)$ (if applicable). For subtree roots we store their hashed surrogate, and for long edges we store their length. We also store the hash functions $\{H_\ell\}$.
\begin{lemma}\label{lmm:sketch_size}
The total sketch size is
\[
  O \left( n\left((d+\log n)\log(1/\epsilon) + \log\log\Phi + \log\frac{q}{\delta}\right) + d\log\Phi \right)
  \;\; \text{bits.}
\]
\end{lemma}
%

As a preprocessing step, Alice can reduce the dimension of her points to $O(\epsilon^{-2}\log(qn/\delta))$ by a Johnson-Lindenstrauss projection.
She then augments the sketch with the projection, in order for Bob to be able to project his points as well.
By~\cite{kane2011almost}, the projection can be stored with $O(\log d + \log(q/\delta)\cdot\log\log((q/\delta)/\epsilon))$ bits.
This yields the sketch size stated in Theorem~\ref{thm:ann_ub}.

\paragraph{Remark} Both the hash functions and the projection map can be sampled using public randomness.
If one is only interested in the communication complexity, one can use the general reduction from public to private randomness due to~\cite{newman1991private}, which replaces the public coins by augmenting $O(\log(nd\Phi))$ bits to the sketch (since Alice's input has size $O(nd\Phi)$ bits).
The bound in~\Cref{thm:ann_ub} then improves to $O\left( n\left(\frac{\log n\cdot\log(1/\epsilon)}{\epsilon^2} + \log\log\Phi + \log\left(\frac{q}{\delta}\right)\right) + \log\Phi \right)$ bits, and the bound in~\Cref{thm:distances_ub} improves to $O\left(\frac{n}{\epsilon^2}\left(\log n\cdot\log(1/\epsilon) + \log(d\Phi)\log\left(\frac{q}{\delta}\right)\right) \right)$ bits.
However, that reduction is non-constructive; we state our bounds so as to describe explicit sketches.

\section{Approximate Nearest Neighbor Search}\label{sec:ann}

We now describe our approximate nearest neighbor search query procedure, and prove~\Cref{thm:ann_ub}.
Suppose Bob wants to report a $(1+\epsilon)$-approximate nearest neighbor in $X$ for a point $y\in Y$.

\paragraph{Algorithm Report Nearest Neighbor:}
\begin{enumerate}
  \item Start at the subtree $T'\in\mathcal F(T)$ that contains the root of $T$.
  \item Recover all surrogates $\{s^*(v):v\in T'\}$, by the subroutine below.
  \item Let $v$ be the leaf of $T'$ that minimizes $\norm{y-s^*(v)}$.
  \item If $v$ is the head of a long edge, recurse on the subtree under that long edge. Otherwise $v$ is a leaf in $T$, and in that case return $c(v)$.
\end{enumerate}

\paragraph{Subroutine Recover Surrogates:}
This is a subroutine that attempts to recover all surrogates $\{s^*(v):v\in T'\}$ in a given subtree $T'\in\mathcal F(T)$, using both Alice's sketch and Bob's point $y$.

Observe that to this end, the only information missing from the sketch is the root surrogate $s^*(r)$, which served as the induction base for defining the rest of the surrogates. The induction steps are fully defined by $\ell(v)$, $\mathrm{in}(v)$, $\gamma(v)$, and $\eta(v)$, which are stored in the sketch for every node $v\neq r$ in the subtree.
The missing root surrogate was defined as $s^*(r)=\mathcal N_{2^{\ell(r)}}[x_{c(r)}]$.
Instead, the sketch stores its hashed value $H_{\ell(r)}(\mathcal N_{2^{\ell(r)}}[x_{c(r)}])$ and the hash function $H_{\ell(r)}$.\footnote{Note that fully storing the root surrogates is prohibitive: $\mathcal N_{2^{\ell(r)}}$ has $\Theta(2\sqrt{d}\Phi/2^{\ell(r)})^d$ cells, hence storing a cell ID takes $\Omega(d\log d)$ bits, and since there can be $\Omega(n)$ subtree roots, this would bring the total sketch size to $\Omega(nd\log d)$.}

The subroutine attempts to reverse the hash.
It enumerates over all points $p\in \mathcal N_{2^{\ell(r)}}$ such that $\norm{p-y}\leq2\cdot2^{\ell(r)}$.
For each $p$ it computes $H_{\ell(r)}(p)$.
If $H_{\ell(r)}(x_{c(r)})=H_{\ell(r)}(p)$ then it sets $s^*(r)=p$ and recovers all surrogates accordingly.
If either no $p$, or more than one $p$, satisfy $H_{\ell(r)}(x_{c(r)})=H_{\ell(r)}(p)$, then it proceeds with $s^*(r)$ set to an arbitrary point (say, the origin in $\R^d$).

\paragraph{Analysis.} Let $r_0,r_1,\ldots$ be the roots of the subtrees traversed on the algorithm.
Note that they reside on a downward path in $T$.

\begin{claim}
$\norm{x_{c(r_0)}-y} \leq 2^{\ell(r_0)}$.
\end{claim}
\begin{proof}
Since $X\cup Y\subset\{-\Phi \ldots \Phi\}^d$, we have $\norm{x_{c(r_0)}-y} \leq 2\sqrt{d}\Phi\leq2^{\lceil\log(2\sqrt d\Phi)\rceil}=2^{\ell(r_0)}$.
\end{proof}

Let $t$ be the smallest such that $r_t$ satisfies $\norm{x_{c(r_t)}-y}>2^{\ell(r_t)}$.
(The algorithm does not identify $t$, but we will use it for the analysis.)

\begin{lemma}\label{lmm:hashes}
With probability $1-\delta/q$, for every $i=0,\ldots,t-1$ simultaneously,
the subroutine recovers $s^*(r_i)$ correctly as $\mathcal N_{2^{\ell(r)}}[x_{c(r)}]$.
(Consequently, all surrogates in the subtree rooted by $r_i$ are also recovered correctly.)
\end{lemma}
\begin{proof}
Fix a subtree $T'\in\mathcal F(T)$ rooted in $r$, that satisfies $\norm{y-x_{c(r)}}\leq2^{\ell(r)}$.
Since $\norm{x_{c(r)}-s^*(r)}\leq2^{\ell(r)}$ (by Lemma~\ref{lmm:surrogates}), we have $\norm{y-s^*(r)}\leq2\cdot2^{\ell(r)}$.
Hence the surrogate recovery subroutine tries $s^*(r)$ as one of the hash pre-image candidates, and will identify that $H_{\ell(r)}(s^*(r))$ matches the hash stored in the sketch.
Furthermore, by~\Cref{clm:gridball}, the number of candidates is at most $O(1)^d$.
Since the range of $H_{\ell(r)}$ has size $m=O(1)^d\cdot\log(2\sqrt d\Phi)\cdot q/\delta$, then with probability $1-\delta/(q\log(2\sqrt d\Phi))$ there are no collisions, and $s^*(r)$ is recovered correctly.
The lemma follows by taking a union bound over the first $t$ subtrees traversed by the algorithm, i.e.~those rooted by $r_i$ for $i=0,1,\ldots,t-1$. Noting that $t$ is upper-bounded by the number of levels in the tree, $\log(2\sqrt d\Phi)$, we get that all the $s^*(r_i)$'s are recovered correctly simultaneously with probability $1-\delta/q$.
\end{proof}

From now on we assume that the event in Lemma~\ref{lmm:hashes} succeeds, meaning in steps $0,1,\ldots,t-1$, the algorithm recovers all surrogates correctly. We henceforth prove that under this event, the algorithm returns a $(1+\epsilon)$-approximate nearest neighbor of $y$.
In what follows, let $x^*\in X$ be a fixed true nearest neighbor of $y$ in $X$.

\begin{lemma}\label{lmm:annrounds}
Let $T'\in\mathcal F(T)$ be a subtree rooted in $r$, such that $x^*\in C(r)$.
Let $v$ a leaf of $T'$ that minimizes $\norm{y-s^*(v)}$.
Then either $x^*\in C(v)$,
or every $z\in C(v)$ is a $(1+O(\epsilon))$-approximate nearest neighbor of $y$.
\end{lemma}
\begin{proof}
Suppose w.l.o.g.~by scaling that $\epsilon<1/6$.
If $x^*\in C(v)$ then we are done. Assume now that $x^*\in C(u)$ for a leaf $u\neq v$ of $T'$.
Let $\ell:=\max\{\ell(v),\ell(u)\}$. We start by showing that $\norm{y-x^*}>\frac{1}{4}\cdot 2^\ell$. Assume by contradiction this is not the case. Since $u$ is a subtree leaf and $x^*\in C(u)$, we have $\norm{x^*-x_{c(u)}}\leq2^{\ell}\epsilon$ by Lemma~\ref{lmm:subtree_leaf}.
We also have $\norm{x_{c(u)}-s^*(u)}\leq2^{\ell}\epsilon$ by Lemma~\ref{lmm:surrogates}. Together, $\norm{y-s^*(u)}\leq(\frac{1}{4}+2\epsilon)2^\ell$. On the other hand, by the triangle inequality,
$\norm{y-s^*(v)} \geq \norm{x^*-x_{c(v)}} - \norm{y-x^*} - \norm{x_{c(v)}-s^*(v)}$.
Noting that $\norm{x^*-x_{c(v)}}\geq2^\ell$ (by Lemma~\ref{lmm:separation}, since $x^*$ and $x_{c(v)}$ are separated at level $\ell$), $\norm{y-x^*}\leq\frac{1}{4}\cdot 2^\ell$ (by the contradiction hypothesis) and $\norm{x_{c(v)}-s^*(v)}\leq2^\ell\epsilon$ (by Lemma~\ref{lmm:surrogates}), we get $\norm{y-s^*(v)}\geq(\frac{3}{4}-\epsilon)2^\ell>(\frac{1}{4}+2\epsilon)2^\ell\geq\norm{y-s^*(u)}$. This contradicts the choice of $v$.

The lemma now follows because for every $z\in C(v)$,
\begin{align}
\norm{y-z} &\leq \norm{y-s^*(v)} + \norm{s^*(v)-x_{c(v)}} + \norm{x_{c(v)}-z} \label{ineq1} \\
&\leq \norm{y-s^*(u)} + \norm{s^*(v)-x_{c(v)}} + \norm{x_{c(v)}-z} \label{ineq2} \\
&\leq \norm{y-x^*} + \norm{x^*-x_{c(u)}} + \norm{x_{c(u)}-s^*(u)} +\norm{s^*(v)-x_{c(v)}} + \norm{x_{c(v)}-z} \label{ineq3} \\
&\leq \norm{y-x^*} + 4\cdot2^\ell\epsilon \label{ineq4} \\
&\leq (1+16\epsilon)\norm{y-x^*}, \label{ineq5}
\end{align}
where~(\ref{ineq1}) and (\ref{ineq3}) are by the triangle inequality, (\ref{ineq2}) is since $\norm{y-s^*(v)}\leq\norm{y-s^*(u)}$ by choice of $v$, (\ref{ineq4}) is by Lemmas~\ref{lmm:subtree_leaf} and~\ref{lmm:surrogates}, and~(\ref{ineq5}) is since we have shown that $\norm{y-x^*}>\frac{1}{4}\cdot 2^\ell$.
Therefore $z$ is a $(1+16\epsilon)$-approximate nearest neighbor of $y$.
\end{proof}

\paragraph{Proof of~\Cref{thm:ann_ub}.}
We may assume w.l.o.g.~that $\epsilon$ is smaller than a sufficiently small constant.
Suppose that the event in Lemma~\ref{lmm:hashes} holds, hence all surrogates in the subtrees rooted by $r_0,r_1,\ldots,r_{t-1}$ are recovered correctly.
We consider two cases. In the first case, $x^*\notin C(r_t)$.
Let $i\in\{1,\ldots,t\}$ be the smallest such that $x^*\notin C(r_i)$.
By applying Lemma~\ref{lmm:annrounds} on $r_{i-1}$, we have that every point in $C(r_i)$ is a $(1+O(\epsilon))$-approximate nearest neighbor of $y$. After reaching $r_i$, the algorithm would return the center of some leaf reachable from $r_i$, and it would be a correct output.

In the second case, $x^*\in C(r_t)$. We will show that every point in $C(r_t)$ is a $(1+O(\epsilon))$-approximate nearest neighbor of $y$, so once again, once the algorithm arrives at $r_t$ it can return anything.
By Lemma~\ref{lmm:subtree_root}, every $x\in C(r_t)$ satisfies
\begin{equation}\label{eq:ann_endgame}
\norm{x-x^*} \leq  2^{\ell(r_t)}\epsilon .
\end{equation}
In particular, $\norm{x_{c(r_t)}-x^*} \leq 2^{\ell(r_t)}\epsilon$.
By definition of $t$ we have $\norm{x_{c(r_t)}-y}>2^{\ell(r_t)}$.
Combining the two yields $\norm{y-x^*} \geq \norm{y-x_{c(r_t)}} - \norm{x_{c(r_t)}-x^*} > (1-\epsilon)2^{\ell(r_t)}$. 
Combining this with~\cref{eq:ann_endgame}, we find that every $x\in C(r_t)$ satisfies $\norm{x-x^*}\leq\frac{\epsilon}{1-\epsilon}\norm{y-x^*}$, and hence $\norm{y-x}\leq(1+2\epsilon)\norm{y-x^*}$ (for $\epsilon\leq1/2$). Hence $x$ is a $(1+2\epsilon)$-nearest neighbor of $y$.

The proof assumes the event in Lemma~\ref{lmm:hashes}, which occurs with probability $1-\delta/q$.
By a union bound, the simultaneous success probability of the $q$ query points of Bob is $1-\delta$ as required. \qed

%

\section{Distance Estimation}\label{sec:dist}
We now prove~\cref{thm:distances_ub}.
To this end, we augment the basic sketch from Section~\ref{sec:sketch} with additional information, relying on the following distance sketches due to~\cite{achlioptas2001database} (following~\cite{johnson1984extensions}) and~\cite{kushilevitz2000efficient}.

\begin{lemma}[\cite{achlioptas2001database}]\label{lmm:binaryjl}
Let $\epsilon,\delta'>0$.
Let $d'=c\epsilon^{-2}\log(1/\delta')$ for a sufficiently large constant $c>0$.
Let $M$ be a random $d'\times d$ matrix in which every entry is chosen independently uniformly at random from $\{-1/\sqrt{d'},1/\sqrt{d'}\}$.
Then for every $x,y\in\R^d$, with probability $1-\delta'$, $\norm{Mx-My}=(1\pm\epsilon)\norm{x-y}$.
\end{lemma}

\begin{lemma}[\cite{kushilevitz2000efficient}]\label{lmm:kor}
Let $R>0$ be fixed and let $\epsilon,\delta'>0$. There is a randomized map $\mathrm{sk}_R$ of vectors in $\R^d$ into $O(\epsilon^{-2}\log(1/\delta'))$ bits, with the following guarantee.
For every $x,y\in\R^d$, given $\mathrm{sk}_R(x)$ and $\mathrm{sk}_R(y)$, one can output the following with probability $1-\delta'$:
\begin{itemize}
  \item If $R\leq\norm{x-y}\leq 2R$, output a $(1+\epsilon)$-estimate of $\norm{x-y}$.
  \item If $\norm{x-y}\leq (1-\epsilon)R$, output ``Small''.
  \item If $\norm{x-y}\geq (1+\epsilon)R$, output ``Large''.  
\end{itemize} 
\end{lemma}

%

We augment the basic sketch from Section~\ref{sec:sketch} as follows.
We sample a matrix $M$ from Lemma~\ref{lmm:binaryjl}, with $\delta'=\delta/q$.
In addition, for every level $\ell$ in the tree $T$, we sample a map $\mathrm{sk}_{2^\ell}$ from Lemma~\ref{lmm:kor}, with $\delta'=\delta/(q\log(2\sqrt{d}\Phi))$.
For every subtree root $r$ in $T$, we store $Mx_{c(r)}$ and $\mathrm{sk}_{2^{\ell(r)}}(x_{c(r)})$ in the sketch.
Let us calculate the added size to the sketch:

\begin{itemize}
  \item Since $x_{c(r)}$ has $d$ coordinates of magnitude $O(\Phi)$ each, $Mx_{c(r)}$ has $d'$ coordinates of magnitude $O(d\Phi)$ each. Since there are $O(n)$ subtree roots (cf.~Lemma~\ref{lmm:sketch_size}), storing $Mx_{c(r)}$ for every $r$ adds $O(nd'd\Phi)=O(\epsilon^{-2}n\log(q/\delta)\log(d\Phi))$ bits to the sketch. In addition we store the matrix $M$, which takes $O(d'd)$ bits to store, which is dominated by the previous term.
  \item By Lemma~\ref{lmm:kor}, each $\mathrm{sk}_{2^{\ell(r)}}(x_{c(r)})$ adds $O(\epsilon^{-2}\log(q\log(2\sqrt{d}\Phi)/\delta))$ bits to the sketch, and as above there are $O(n)$ of these. In addition we store the map $\mathrm{sk}_{2^{\ell(r)}}$ for every $\ell$. Each map takes $\mathrm{poly}(d,\log\Phi,\log(q/\delta),1/\epsilon)$ bits to store.
\end{itemize}
In total, we get the sketch size stated in~\Cref{thm:distances_ub}.
Next we show how to compute all distances from a new query point $y$.

\paragraph{Query algorithm.}
Given the sketch, an index $k\in[n]$ of a point in $X$, and a new query point $y$, the algorithm needs to estimate $\norm{y-x_k}$ up to $1\pm O(\epsilon)$ distortion. It proceeds as follows.
\begin{enumerate}
  \item Perform the approximate nearest neighbor query algorithm from Section~\ref{sec:ann}. Let $r_0,r_1,\ldots$ be the downward sequence of subtree roots traversed by it.
  \item For each $r_j$, estimate from the sketch whether $\norm{y-x_{c(r_j)}}\leq 2^{\ell(r_j)}$.
  This can be done by Lemma~\ref{lmm:kor}, since the sketch stores $\mathrm{sk}_{2^{\ell(r_j)}}(x_{c(r_j)})$ and also the map $\mathrm{sk}_{2^{\ell(r_j)}}$, with which we can compute $\mathrm{sk}_{2^{\ell(r_j)}}(y)$.
  
   \item Let $t$ be the smallest $j$ that satisfies $\norm{y-x_{c(r_j)}} > 2^{\ell(r_j)}$ according the estimates of Lemma~\ref{lmm:kor}.
  (This attempts to recover from the sketch the same $t$ as defined in the analysis in Section~\ref{sec:ann}.)
  \item Let $t_k\in\{0,\ldots,t\}$ be the maximal such that $x_k\in C(r_{t_k})$.
  
  (In words, $r_{t_k}$ is the root of the subtree in which $x_k$ and $y$ ``part ways''.)
  \item If $t_k=t$, return $\norm{My-Mx_{c(r_t)}}$. Note that $M$ and $Mx_{c(r_t)}$ are stored in the sketch.
  \item If $t_k<t$, let $v_k$ be the bottom node on the downward path from $r_{t_k}$ to $\mathrm{leaf}(x_k)$ that does not traverse a long edge. Return $\norm{y-s^*(v_k)}$.
\end{enumerate}

\paragraph{Analysis.}
Fix a query point $y$.
Define the ``good event'' $\mathcal A(y)$ as the intersection of the following:
\begin{enumerate}
  \item For every subtree root $r_j$ traversed by the query algorithm above, the invocation of Lemma~\ref{lmm:kor} on $\mathrm{sk}_{2^{\ell(r_j)}}(x_{c(r_j)})$ and $\mathrm{sk}_{2^{\ell(r_j)}}(y)$ succeeds in deciding whether $\norm{y-x_{c(r_j)}}\leq 2^{\ell(r_j)}$.
  Specifically, this ensures that $\norm{y-x_{c(r_j)}}\leq 2^{\ell(r_j)}$ for every $j<t$, and $\norm{y-x_{c(r_t)}}\geq (1-\epsilon)2^{\ell(r_t)}$.
  Recalling that we invoked the lemma with $\delta'=\delta/(q\log(2\sqrt{d}\Phi))$, we can take a union bound and succeed in all levels simultaneously with probability $1-\delta/q$.

  \item $\norm{My-Mx_{c(r_t)}}=(1\pm\epsilon)\norm{y-x_{c(r_t)}}$. By Lemma~\ref{lmm:binaryjl} this holds with probability $1-\delta/q$.
\end{enumerate}
Altogether, $\mathcal A(y)$ occurs with probability $1-O(\delta/q)$.

\begin{lemma}
Conditioned on $\mathcal A(y)$ occuring, with probabiliy $1-\delta/q$, Lemma~\ref{lmm:hashes} holds. Namely, the query algorithm correctly recovers all surrogrates in the subtrees rooted by $r_j$ for $j=0,1,\ldots,t-1$.
\end{lemma}
\begin{proof}
The proof of Lemma~\ref{lmm:hashes} in Section~\ref{sec:ann} relied on having $\norm{y-x_{c(r_j)}}\leq2^{\ell(r_j)}$ for every $j<t$. Conditioning on $\mathcal A(y)$ ensures this holds.
\end{proof}

\paragraph{Proof of~\Cref{thm:distances_ub}.}
Let $\mathcal A^*(y)$ denote the event in which both $\mathcal A(y)$ occurs and the conclusion of Lemma~\ref{lmm:hashes} occurs. By the above lemma, $\mathcal A^*(y)$ happens with probability $1-O(\delta/q)$.
From now on we will assume that $\mathcal A^*(y)$ occurs, and conditioned on this, we will show that the distance from $y$ to any data point can be deterministically estimated correctly.
To this end, fix $k\in[n]$ and suppose our goal is to estimate $\norm{y-x_k}$. Let $t_k$ and $v_k$ be as defined by the distance query algorithm above.
We handle the two cases of the algorithm separately.

\textbf{Case I:} $t_k=t$. This means $x_k\in C(r_t)$. By Lemma~\ref{lmm:subtree_root} we have $\norm{x_k-x_{c(r_t)}}\leq2^{\ell(r_t)}\epsilon$. By the occurance of $\mathcal A^*(y)$ we have $\norm{y-x_{c(r_t)}}>(1-\epsilon)2^{\ell(r_t)}$. Together,
  $\norm{y-x_k} =
  \norm{y-x_{c(r_t)}} \pm \norm{x_k-x_{c(r_t)}} =
  (1\pm2\epsilon)\norm{y-x_{c(r_t)}}$.
This means that $\norm{y-x_{c(r_t)}}$ is a good estimate for $\norm{y-x_k}$. Since $\mathcal A^*(y)$ occurs, it holds that $\norm{My-Mx_{c(r_t)}}=(1\pm\epsilon)\norm{y-x_{c(r_t)}}$, hence $\norm{My-Mx_{c(r_t)}}$ is also a good estimate for $\norm{y-x_k}$, and this is what the algorithm returns.

\textbf{Case II:} $t_k<t$. Let $T_{t_k}$ be the subtree rooted by $r_{t_k}$. 
By the occurance of $\mathcal A^*(y)$, all surrogates in $T_{t_k}$ are recovered correctly, and in particular $s^*(v_k)$ is recovered correctly.
By Lemma~\ref{lmm:surrogates} we have $\norm{x_{c(v_k)}-s^*(v_k)}ֿ\leq2^{\ell(v_k)}\epsilon$,
and by Lemma~\ref{lmm:subtree_leaf} (noting that $x_k\in C(v_k)$ by choice of $v_k$) we have $\norm{x_k-x_{c(v_k)}}\leq2^{\ell(v_k)}\epsilon$. Together,
$\norm{x_k-s^*(v_k)}ֿ\leq2\cdot2^{\ell(v_k)}\epsilon$.

Let $v$ be the leaf in $T_{t_k}$ that minimizes $\norm{y-s^*(v)}$ (over all leaves of $T_{t_k}$). Equivalently, $v$ is the top node of the long edge whose bottom node is $r_{t_k+1}$.
Let $\ell:=\max\{\ell(v),\ell(v_k)\}$.
By choice of $t_k$ we have $v\neq v_k$, hence the centers of these two leaves are separated already at level $\ell$, hence $\norm{x_{c(v_k)}-x_{c(v)}}\geq2^{\ell}$ by Lemma~\ref{lmm:separation}.
By two applications of Lemma~\ref{lmm:surrogates} we have $\norm{x_{c(v_k)}-s^*(v_k)}ֿ\leq2^{\ell}\epsilon$ and $\norm{x_{c(v)}-s^*(v)}ֿ\leq2^{\ell}\epsilon$.
Together, $\norm{s^*(v_k)-s^*(v)}\geq(1-2\epsilon)\cdot2^{\ell}$.
Since $y$ is closer to $s^*(v)$ than to $s^*(v_k)$ (by choice of $v$), we have
$\norm{y-s^*(v_k)} \geq \frac12\cdot\norm{s^*(v_k)-s^*(v)} \geq \left(\frac12-\epsilon\right)\cdot2^{\ell}$.
Combining this with $\norm{x_k-s^*(v_k)}ֿ\leq2\cdot2^{\ell(v_k)}\epsilon$, which was shown above, yields
$\norm{x_k-s^*(v_k)}ֿ\leq\epsilon\cdot\frac{1}{1/2-\epsilon}\cdot\norm{y-s^*(v_k)}=O(\epsilon)\cdot\norm{y-s^*(v_k)}$.
Therefore,
$\norm{y-x_k} = \norm{y-s^*(v_k)} \pm \norm{x_k-s^*(v_k)} = (1\pm O(\epsilon))\cdot\norm{y-s^*(v_k)}$,
which means $\norm{y-s^*(v_k)}$ is a good estimate for $\norm{y-x_k}$, and this is what the algorithm returns.

\textbf{Conclusion:} Combining both cases, we have shown that for any query point $y$, all distances from $y$ to $X$ can be estimated correctly with probability $1-O(\delta/q)$. Taking a union bound over $q$ queries, and scaling $\delta$ and $\epsilon$ appropriately by a constant, yields the theorem. \qed

\paragraph{Acknowledgments}
This work was supported by grants from the MITEI-Shell program, Amazon Research Award and Simons Investigator Award.

\bibliographystyle{amsalpha}
\bibliography{sublinear_space_nn}

\newpage
\appendix
\section{Deferred Proofs from Section~\ref{sec:sketch}}\label{sec:sketch_proofs}

\begin{proof}[Proof of Lemma~\ref{lmm:tree_size}]
Charging the degree-$1$ nodes along every maximal $1$-path to its bottom (non-degree-$1$) node,
the total number of nodes after top-out compression is bounded by
\[ \sum_{v:\mathrm{deg}(v)\neq1}\Lambda(v) . \]
\cite{indyk2017near} show this is at most $O(n\log(1/\epsilon))$.
The difference is that their compression replaces summands larger than $\Lambda(v)$ by zero, while our (top-out) compression trims them to $\Lambda(v)$.
\end{proof}

\begin{proof}[Proof of Lemma~\ref{lmm:subtree_root}]
By top-out compression, $v$ is the top of a downward $1$-path of length $\Lambda(v')$ whose bottom node is $v'$. Since no clusters are joined along a $1$-path, we have $C(v')=C(v)$, hence $x,x'\in C(v')$ and hence $\norm{x-x'} \leq \Delta(v')$. 
Noting that $\ell(v)=\ell(v')+\Lambda(v')=\ell(v')+\log(\Delta(v')/(2^{\ell(v')}\epsilon))=\log(\Delta(v')/\epsilon)$ and rearranging, we find $\Delta(v')=2^{\ell(v)}\epsilon$, which yields the claim.
\end{proof}

\begin{proof}[Proof of Lemma~\ref{lmm:subtree_leaf}]
If $u$ is a leaf in $T$ then $C(u)$ is a singleton cluster, hence $x=x'$.
Otherwise $u$ is the top node of a long edge, and the claim follows by Lemma~\ref{lmm:subtree_root} on the bottom node of that long edge.
\end{proof}

\begin{proof}[Proof of Lemma~\ref{lmm:surrogates}]
The first part of the lemma (where $v$ is any node, not necessarily a subtree leaf) is proved by induction on the ingresses.
In the base case we use that $\norm{x_{c(v)}-s^*(v)}\leq2^{\ell(v)}$ by the choice of grid net.
The induction step is identical to~\cite{indyk2017near}.
The ``furthermore'' part of the lemma then follows as a corollary due to the refined definition of $\gamma(v)$ for a subtree leaf $v$, in the induction step leading to it.
\end{proof}

\begin{proof}[Proof of Lemma~\ref{lmm:sketch_size}]
The sketch of~\cite{indyk2017near} stores the compressed tree $T'$, with each node annotated by its center $c(v)$, ingress $\mathrm{in}(v)$, precision $\gamma(v)$ and quantized displacement $\eta(v)$.
Every long edge is annotated by its length.
They show this takes
\[ O\left( n\left((d+\log n)\log(1/\epsilon) + \log\log\Phi\right)\right) \]
bits;
note that by Lemma~\ref{lmm:tree_size}, top-out compression did not effect this bound.

We additionally store the hashed surrogates of subtree roots.
There are $O(n)$ subtrees,\footnote{By construction, the tree of subtrees in $T$ has no degree-$1$ nodes. Since $T$ has $n$ leaves, there are at most $2n-1$ subtrees.} and each hash takes $\log m$ bits to store, which adds $O(n(d+\log\log\Phi+\log(q/\delta)))$ bits to the above.
Finally, we store the hash functions $H_\ell$ for every $\ell$.
The domain of each $H_\ell$ is $N_{2^\ell}$, which is a subset of $\{-\Phi \ldots \Phi\}^d$, and hence $H_\ell$ can be specified by $O(\log(\Phi^d))$ random bits (\cite{carter1979universal}).
Since we do not require independence between hash functions of different levels, we can use the same random bits for all hash functions, adding a total of $O(d\log\Phi)$ bits to the sketch.
\end{proof}

\section{Approximate Nearest Neighbor Sketching Lower Bound}\label{sec:nnlower}

\begin{theorem}\label{thm:ann_lb}
Suppose that $d\geq\Omega(\epsilon^{-2}\log n)$, $\Phi\geq1/\epsilon$, and $1/n^{0.5-\beta}\leq\epsilon\leq\epsilon_0$ for a constant $\beta>0$ and a sufficiently small constant $\epsilon_0$.
Suppose also that $\delta<1/n^2$.
Then, for the all-nearest-neighbors problem, Alice must use a sketch of at least $\Omega(\beta\epsilon^{-2}n\log n)$ bits.
\end{theorem}
\begin{proof}
We start with dimension $d=n+1+\log n$; it can then be reduced by standard dimension reduction.
Fix $k=1/\epsilon^2$ and assume w.l.o.g.~that $k$ is a square integer (by taking $\epsilon$ to be appropriately small). Note that since $\epsilon>1/\sqrt{n}$ we have $k\leq n$, and that since $\Phi\geq1/\epsilon$ we have $\sqrt{k}\leq\Phi$. 

The data set will consist of $2n$ points, $x_1,\ldots,x_n$ and $z_1,\ldots,z_n$.
Let $i\in[n]$. We choose the first $n$ coordinates of $x_i$ to be an arbitrary $k$-sparse vector, in which each nonzero coordinate equals $1/\sqrt{k}$. Note that the norm of this part is $1$. The $(n+1)$th coordinate of $x_i$ is set to $0$. The remaining $\log n$ coordinates encode the binary encoding of $i$, with each coordinate multiplied by $10$.

Next we define $z_i$. The first $n$ coordinates are $0$. The $(n+1)$th coordinate equals $\sqrt{1-\epsilon}$. The remaining $\log n$ coordinates encode $i$ similarly to $x_i$.

The number of different choices for $\{x_1,\ldots,x_n\}$ is ${n\choose k}^n$.
Therefore if we show that one can fully recover $x_1,\ldots,x_n$ from a given all-nearest-neighbor sketch of the dataset, we would get the desired lower bound
\[
  \log\left({n\choose k}^n\right) \geq
   nk\log(n/k) = \epsilon^{-2}n\log(\epsilon^2n) =\epsilon^{-2}n\log(n^{2\beta}) = 2\beta\epsilon^{-2}n\log n .
\]

Suppose we have such a sketch.
For given $i,j\in[n]$ we now show how to recover the $j$th coordinate of $x_i$, denoted $x_i(j)$, with a single approximate nearest neighbor query.
Let $y_{ij}$ be the following vector in $\R^d$: The first $n+1$ coordinates are all zeros, except for the $j$th coordinate which is set to $1$. The last $\log n$ coordinates encode $i$ similary to $x_i$ and $z_i$.

Consider the distances from $y_{ij}$ to all data points. We start with $x_i$. It is identical to $y_i$ in the last $\log n+1$ coordinates, so we will restrict both to the first $n$ coordinates and denote the restricted vectors by $x_i^{:n}$ and $y_{ij}^{:n}$. $x_i^{:n}$ is a $k$-sparse vector with nonzero entries equal to $1/\sqrt{k}$, hence $\norm{x_i^{:n}}=1$. $y_{ij}^{:n}$ is just the standard basis vector $e_j$ in $\R^n$. Hence,
\[
  \norm{x_i-y_{ij}}^2 =
  \norm{x_i^{:n}-y_{ij}^{:n}}^2 =
  \norm{x_i^{:n}}^2 + \norm{y_{ij}^{:n}}^2 - 2(x_i^{:n})^\top y_{ij}^{:n} =
  2-2x_i(j).
\]
This equals $2$ if $x_i(j)=0$ and $2-2/\sqrt k=2-2\epsilon$ if $x_i(j)=1/\sqrt k$.

Next consider $z_i$. It is identical to $y_{ij}$ in all except the $j$th coordinate, which is $0$ in $z_i$ and $1$ in $y_{ij}$, and the $(n+1)$th coordinate, which is $0$ for $y_{ij}$ and $\sqrt{1-\epsilon}$ for $z_i$. Therefore, $\norm{z_i-y_{ij}}^2=2-\epsilon$.

Finally, for every $i'\neq i$, both $x_{i'}$ and $z_{i'}$ are at distance at least $10$ from $y_{ij}$ due to the encoding of $i$ (as binary multiplied by $10$) in the last $\log n$ coordinates.

In summation we have established the following:
\begin{itemize}
  \item If $x_i(j)\neq0$, then the closest point to $y_{ij}$ in the dataset is $x_i$ at distance $\sqrt{2-2\epsilon}$, and the next closest point is $z_i$ at distance $\sqrt{2-\epsilon}$.
  \item If $x_i(j)=0$, then the closest point to $y_{ij}$ in the dataset is $z_i$ at distance $\sqrt{2-\epsilon}$, and the next closest point is $x_i$ at distance $2$.
\end{itemize}
Therefore, if the sketch supports $(1+\frac{1}{8}\epsilon)$-approximate nearest neighbors, we can recover the true nearest neighbor of $y_{ij}$ and thus recover $x_i(j)$. By hypothesis, the query succeeds with probability $\delta<1/n^2$.
By a union bound over all $i,j\in[n]$ we can recover all of $x_1,\ldots,x_n$ simultaneously, and the theorem follows.
\end{proof}

\section{Lower Bound for Distance Estimation}\label{sec:dist_lb}
In this section we prove~\Cref{thm:distances_lb}.
We handle the two terms in the lower bound separately.

\subsection{First Lower Bound Term}
\begin{lemma}\label{lmm:lb1}
Suppose that $d\geq\Omega(\epsilon^{-2}\log n)$; $\Phi\geq1/\epsilon$; $\epsilon$ is at most a sufficiently small constant; and $\epsilon\geq1/n^{0.5-\rho'}$ for a constant $\rho'>0$.
Then, for the all-cross-distances problem, Alice must use a sketch of at least $\Omega(\rho'\epsilon^{-2}n\log n)$ bits.
\end{lemma}
\begin{proof}
Consider the following problem: Given a dataset $X\subset\{-\Phi \ldots \Phi\}^d$ consisting of $n$ points, we need to produce a sketch, from which we can recover (deterministically) all distances $\{\norm{x-x'}:x,x'\in X\}$ up to distortion $1\pm\epsilon$.
This is the problem considered in~\cite{indyk2017near}, and they show a lower bound of $\Omega(\rho'\epsilon^{-2}n\log n)$ under the asssumptions of the current lemma. We will obtain the same lower bound by reducing this problem to all-cross-distances.

To this end, suppose we have a given sketching procedure for the all-cross-distances problem that uses $s=s(n,d,\Phi,1,\epsilon,\delta)$ amortized bits per point.
We invoke it on $X$ and denote the resulting sketch by $S_0$. For every point $y\in\{-\Phi \ldots \Phi\}^d$, with probabiliy $1-\delta$, all distances $\{\norm{x-y}:x\in X\}$ can be recovered from $S_0$. In particular, this holds in expectation for $(1-\delta)n$ of the points in $X$.
By Markov's inequality, this holds for $\frac12(1-\delta)n>\frac14n$ of the points in $X$ with probability at least $1/2$.
We proceed by recursion on the remaining $\frac34n$ points in $X$.
The sketch produced in the $i$th step of the recursion is denoted by $S_i$ and has total size $(\frac34)^ins$ bits.
After $t=O(\log n)$ steps, with nonzero probability $1/n^{O(1)}$, we have produced a sequence of sketches $S_0,\ldots,S_t$ from which every distance in $\{\norm{x-x'}:x,x'\in X\}$ can be recovered, with a total size of $O(\sum_{i=0}^t(\frac34)^ins)=O(ns)$ bits. This yields the desired lower bound.
\end{proof}

\subsection{Second Lower Bound Term}

\begin{lemma}\label{lmm:lb2}
Suppose that $d^{1-\rho}\geq\epsilon^{-2}\log(q/\delta)$ for a constant $\rho>0$, and $\epsilon$ is at most a sufficiently small constant.
Then, for the all-cross-distances problem, Alice must use a sketch of at least $\Omega(\epsilon^{-2}n\log(d\Phi)\log(q/\delta))$ bits.
\end{lemma}

The proof is by adapting the framework of~\cite{molinaro2013beating}, who proved (among other results) this statement for the case $q=n$.
We describe the adaption and refer to~\cite{molinaro2013beating} for missing details that remain similar.

\subsubsection{Preliminaries}
We follow the approach of~\cite{jayram2013optimal}, of proving one-way communication lower bounds by reduction to variants of the augmented indexing problem, defined next.
\begin{definition}[Augmented Indexing]\label{def:augind}
In the Augmented Indexing problem $AugInd(k,\delta)$, Alice gets a vector $A$ with $k$ entries, whose elements are entries of a universe of size $20/\delta$.
Bob gets an index $i\in[k]$, an element $e$, and the elements $A(i')$ for every $i'<i$.
Bob needs to decide whether $e=A(i)$, and succeed with probability $1-\delta$.
\end{definition}

\cite{jayram2013optimal} give a one-way communication lower bound of $\Omega(k\log(1/\delta))$ for this problem.
The main component in~\cite{molinaro2013beating} is a modified one-way communication model, in which the protocol is allowed to abort with a substantially larger (constant) probability than it is allowed to err.
We will refer to it simply as the~\emph{abortion model} and refer to~\cite{molinaro2013beating} for the exact definition (which we will not require).
They prove the same lower bound for Augmented Indexing.
\begin{lemma}[informal]\label{lmm:augind}
In the abortion model, the one-way communication complexity of $AugInd(k,\delta)$ is $\Omega(k\log(1/\delta))$.
\end{lemma}

\subsubsection{Variants of Augmented Indexing}

We start by defining a variant of augmented indexing that will be suitable for out purpose.

\begin{definition}
In the Matrix Augmented Indexing problem $MatAugInd(k,m,\delta)$, Alice gets a matrix $A$ of order $k\times m$, whose entries are elements of a universe of size $1/\delta$.
Bob gets indices $i\in[k]$ and $j\in[m]$, an element $e$, and the elements $A(i,j')$ for every $j<j'$.
Bob needs to decide whether $e=A(i,j)$, and succeed with probability $1-\delta$.
\end{definition}

This problem is clearly at least as difficult as $AugInd(km,\delta)$ from Definition~\ref{def:augind}, since in the latter Bob gets more information (namely, if we arrange the vector $A$ in $AugInd(km,\delta)$ as a $k\times m$ matrix, then 
Bob gets all entries of $A$ which lexicographically precede $A(i,j)$).
We get the following immediate corollary from Lemma~\ref{lmm:augind}.

\begin{corollary}\label{cor:mataugind}
In the abortion model,
the one-way communication complexity of $MatAugInd(k,m,\delta)$ is $\Omega(km\log(1/\delta))$.
\end{corollary}

\cite{molinaro2013beating} reformulate Augmented Indexing so that Alice's input is a set instead of vector.
Similary, we reformulate Matrix Augmented Indexing as follows.

\begin{definition}
Let $m>0$ and $k>0$ be integers, and $\delta\in(0,1)$.
Partition the interval $[m/\delta]$ into $m$ intervals $I_1,\ldots,I_m$ of size $1/\delta$ each.

In the Augmented Set List problem $AugSetList(k,m,\delta)$, Alice gets a list of subsets $S_1,\ldots,S_k\subset [m/\delta]$, such that each $S_i$ has size exactly $m$ and contains exactly one element from each interval $I_1,\ldots,I_m$.
Bob gets an index $i\in[k]$, an element $e\in[m/\delta]$ and a subset $T$ of $S_i$ that contains exactly the elements of $S_i$ that are smaller than $e$.
Bob needs to decide whether $e\in S_i$, and succeed with probability at least $1-\delta$.
\end{definition}

The equivalence to Matrix Augmented Indexing is not hard to show; the details are similar to~\cite{molinaro2013beating} and we omit them here. By the equivalence, we get the following corollary from Corollary~\ref{cor:mataugind}.
\begin{corollary}\label{cor:setlist}
In the abortion model,
the one-way communication complexity of $AugSetList(k,m,\delta)$ is $\Omega(km\log(1/\delta))$.
\end{corollary}

Next we define the $q$-fold version of the same problem.
\begin{definition}\label{def:qaugsetlist}
In the problem $q$-$AugSetList(k,m,\delta)$, Alice and Bob get $q$ instances of AugSetList$(k,m,\delta/q)$, and Bob needs to answer correctly on all of them simoultaneously with probability at least $1-\delta$.
\end{definition}

The main tehcnical result of~\cite{molinaro2013beating} is, loosely speaking, a direct-sum theorem which lifts a lower bound in the abortion model to a $q$-fold lower bound in the usual model.
Applying their theorem to Corollary~\ref{cor:setlist}, we obtain the following.
\begin{corollary}\label{cor:augsetlist}
The one-way communication complexity of $q$-$AugSetList(k,m,\delta)$ is $\Omega(qkm\log(q/\delta))$.
\end{corollary}

Finally, we construct a ``generalized augmented indexing'' problem over $r$ copies of the above problem.
\begin{definition}
In the problem $r$-$Ind(q$-$AugSetList(k,m,\delta))$,
Alice gets $r$ instances $A_1,\ldots,A_r$ of $q$-$AugSetList(k,m,\delta)$.
Bob gets an index $j\in[r]$, his part $B_j$ of instance $j$, and Alice's instances $A_1,\ldots,A_{j-1}$.
Bob needs to solve instance $j$ with success probability at least $1-\delta$.
\end{definition}

By standard direct sum results in communication complexity (reproduced in~\cite{molinaro2013beating}) we obtain from Corollary~\ref{cor:augsetlist} the final lower bound we need.
\begin{proposition}\label{prp:lb}
The one-way communication complexity of $r$-$Ind(q$-$AugSetList(k,m,\delta))$ is $\Omega(rqkm\log(q/\delta))$.
\end{proposition}

\subsubsection{Reductions to All-Cross-Distances}

We now prove Lemma~\ref{lmm:lb2}, by reducing $r$-$Ind(q$-$AugSetList(k,m,\delta))$ to the all-cross-distances problem.
We will use two reductions, to get a lower bound once in terms of $d$ and once in terms of $\Phi$.
Specifically, in the first reduction
we will set $m=1/\epsilon^2$, $k=n/q$ and $r=\rho\log d$ (where $\rho$ is the constant from the statement of Lemma~\ref{lmm:lb2}). Then the lower bound we would get by Proposition~\ref{prp:lb} is $\Omega\left(\epsilon^{-2}n\log d\log(q/\delta)\right)$.
In the second reduction we will set $r=\rho\log\Phi$, yielding the lower bound $\Omega\left(\epsilon^{-2}n\log\Phi\log(q/\delta)\right)$.
Together they lead to Lemma~\ref{lmm:lb2}.

In both settings, recall we are reducing to the following problem: For dimension $d=\Omega(\epsilon^{-2}\log(q/\delta))$ and aspect ratio $\Phi$, Alice gets $n$ points, Bob gets $q$ points, and Bob needs to estimate all cross-distances up to distortion $1\pm\epsilon$.

Consider an instance of $r$-$Ind(q$-$AugSetList(k,m,\delta))$.
It can be visualized as follows: Alice gets a matrix $S$ with $n=qk$ rows and $r$ columns, where each entry contains a set of size $m$.
Bob gets an index $j\in[r]$, indices $i_1,\ldots,i_q\in[k]$, elements $e_1,\ldots,e_q$, subsets $T_1\subset S(i_1,j),\ldots, T_q\subset S(i_q,j)$, and the first $j-1$ columns of the matrix $S$.

We now use the encoding scheme of~\cite{jayram2013optimal}, in the set formulation which was given in~\cite{molinaro2013beating}. We restate the result.
\begin{lemma}[\cite{jayram2013optimal}]\label{lmm:jw}
Let $m=1/\epsilon^2$ and $0<\eta<1$. Suppose we have the following setting:
\begin{itemize}
  \item Alice has subsets $S_1,\ldots,S_r$ of $[m/\eta]$.
  \item Bob has an index $j\in[r]$, an element $e\in[m/\eta]$, the subset $T\subset S_j$ of elements smaller than $e$, and the sets $S_1,\ldots,S_{j-1}$.
\end{itemize}
There is a shared-randomness mapping of their inputs into points $v_A,v_B$ and a scale $\Psi>0$ (the scale is known to both), such that
\begin{enumerate}
  \item $v_A,v_B\in\B^D$ for $D=O(\epsilon^{-2}\log(\frac1\eta)\exp(r))$.
  \item If $e\in S_j$ (YES instance) then w.p.~$1-\eta$, $\norm{v_A-v_B}^2 \leq (1-2\epsilon)\Psi$.
  \item If $e\notin S_j$ (NO instance) then w.p.~$1-\eta$, $\norm{v_A-v_B}^2 \geq (1-\epsilon)\Psi$
\end{enumerate}
\end{lemma}

\subsubsection{Lower Bound in terms of $d$}
We start with the first reduction that yields a lower bound in terms of $d$.
\begin{lemma}\label{lmm:lb2a}
Under the assumptions of Lemma~\ref{lmm:lb2}, for the all-cross-distances problem, Alice must use a sketch of at least $\Omega(\epsilon^{-2}n\log(d)\log(q/\delta))$ bits.
\end{lemma}

\begin{proof}
We invoke Lemma~\ref{lmm:jw} with $r=\rho\log d$ and $\eta=\delta/q$.
Note that the latter is the desired success probability in each instance of $q$-$AugSetList(k,m,\delta)$ (cf.~Definition~\ref{def:qaugsetlist}).
Alice encodes each row of the matrix, $(S(i,1),\ldots,S(i,r))$, into a point $x_i$, thus $n$ points $x_1,\ldots,x_n$.
Bob encodes $(S(i,1),\ldots,S(i_z,j-1),T_i,j,e_z)$ for each $z\in[q]$ into a point $y_z$, thus $q$ points $y_1,\ldots,y_z$.
For every $z\in[q]$, the problem represented by row $i_z$ in the matrix $S$ is reduced by Lemma~\ref{lmm:jw} to estimating the distance $\norm{x_{i_z}-y_z}$.
By Item 1 of Lemma~\ref{lmm:jw}, the points $\{x_i\}_{i\in[n]}$, $\{y_z\}_{z\in[q]}$ have binary coordinates and dimension $D=O(\epsilon^{-2}\log(q/\delta)d^\rho)$. 
By the hypothesis $d^{1-\rho}\geq\epsilon^{-2}\log(q/\delta)$ of Lemma~\ref{lmm:lb2}, $D=O(d)$.
Therefore Alice and Bob can now feed them into a given black-box solution of the all-cross-distances problem, which estimates all the required distances and solves $r$-$Ind(q$-$AugSetList(k,m,\delta))$.

Let us establish the success probability of the reduction.
Since we set $\eta=\delta/q$ in Lemma~\ref{lmm:jw}, it preserves each distance $\norm{x_{i_z}-y_z}$ for $z\in[q]$ with probability $1-\delta/q$. By a union bound, it preserves all of them simultaneously with probability $1-\delta$.
The success probability of the all-cross-distances problem, simultaneously on all query points $\{\tilde y_z:z\in[q]\}$, is again $1-\delta$. Altogether, the reduction succeeds with probability $1-O(\delta)$. As a result, the all-cross-distances problem solves the given instance of  $r$-$Ind(q$-$AugSetList(k,m,\delta))$, and Lemma~\ref{lmm:lb2a} follows.
\end{proof}

\subsubsection{Lower Bound in terms of $\Phi$}
We proceed to the second reduction that would yield a lower bound in terms of $\Phi$.
\begin{lemma}\label{lmm:lb2b}
Under the assumptions of Lemma~\ref{lmm:lb2}, for the all-cross-distances problem, Alice must use a sketch of at least $\Omega(\epsilon^{-2}n\log(\Phi)\log(q/\delta))$ bits.
\end{lemma}
\begin{proof}
We may assume that $\Phi\geq d$ since otherwise Lemma~\ref{lmm:lb2b} already follows from Lemma~\ref{lmm:lb2a}. Therefore $\Phi^{1-\rho}\geq\epsilon^{-2}\log(q/\delta)$.

The reduction is very similar to the one in Lemma~\ref{lmm:lb2a}.
Again we evoke Lemma~\ref{lmm:jw} with $\eta=\delta/q$, but this time we set $r=\rho\log\Phi$.
Again we denote Alice's encoded points by $x_1,\ldots ,x_n$, and Bob's by $y_1,\ldots,y_q$.
By Item 1 of Lemma~\ref{lmm:jw}, the points have binary coordinates and dimension $D=O(\epsilon^{-2}\log(q/\delta)\Phi^\rho)$. The difference from Lemma~\ref{lmm:lb2a} is that since it is possible that $\Phi\gg d$, the dimension $D$ is too large for the given black-box solution of the all-cross-distances problem (which is limited to dimension $O(d)$).

To solve this, Alice and Bob project their points into dimension $D'=O(\epsilon^{-2}\log(q/\delta))$ by a Johnson-Lindenstrauss transform, using shared randomness.
Let $\tilde x_1,\ldots,\tilde x_n$ and $\tilde y_1,\ldots,\tilde y_z$ denote the projected points.
After the projection each coordinate has magnitude at most $O(\epsilon^{-2}\log(q/\delta)\Phi^\rho)$.
By our assumption $\Phi^{1-\rho}\geq\Omega(\epsilon^{-2}\log(q/\delta))$, this is at most $O(\Phi)$.
Since the dimension $D'$ is $O(d)$, Alice and Bob can now feed $\tilde x_1,\ldots,\tilde x_n$ and $\tilde y_1,\ldots,\tilde y_z$ into a given black-box solution of the all-cross-distances problem with dimension $O(d)$ and aspect ratio $O(\Phi)$.

Let us establish the success probability of the reduction.
As before, Lemma~\ref{lmm:jw} preserves all the required distances, $\norm{x_{i_z}-y_z}$ for $z\in[q]$, with probability $1-\delta$.
The Johnson-Lindenstrauss transform into dimension $D'$ preserves each distance as $\norm{\tilde x_{i_z}-\tilde y_z}$ with probability at least $1-\delta$, since we picked the dimension to be $D'=O(\epsilon^{-2}\log(\frac{q}{\delta}))$.
The success probability of the all-cross-distances problem simultaneously is again $1-\delta$. Altogether, the reduction succeeds with probability $1-O(\delta)$. As a result, the all-cross-distances problem solves the given instance of  $r$-$Ind(q$-$AugSetList(k,m,\delta))$, and Lemma~\ref{lmm:lb2b} follows.
\end{proof}

\subsubsection{Conclusion}
Lemmas~\ref{lmm:lb2a} and~\ref{lmm:lb2b} together imply Lemma~\ref{lmm:lb2}.
The latter, together with Lemma~\ref{lmm:lb1}, implies~\Cref{thm:distances_lb}. \qed

\section{Practical Variant}\label{sec:middleout}
\cite{indyk2017practical} presented a simplified version of the sketch of~\cite{indyk2017near}, which is lossier by a factor $O(\log\log n)$ in the size bound (more precisely it uses $O(\epsilon^{-2}\log(n)(\log\log(n) + \log(1/\epsilon))+\log\log\Phi)$ bits per point; compare this to Table~\ref{tbl:sketches_related_work}), but on the other hand is practical to implement and was shown to work well empirically.
Both variants do not provably support out-of-sample queries.

In the main part of this work, we showed how to adapt the framework of~\cite{indyk2017near} to support out-of-sample queries with nearly optimal size bounds.
The goal of this section is to show that our techniques can also be applied in a simplified way to~\cite{indyk2017practical} in order to obtain a~\emph{practical} algorithm.
Specifically, focusing on the all-nearest-neighbors problem, we will show that a slight modification to~\cite{indyk2017practical} yields provable support in out-of-sample approximate nearest neighbor queries, with a size bound that is the same as in Theorem~\ref{thm:ann_ub} plus an additive $O(\epsilon^{-2}\log(n)\log\log(n))$ term.

\paragraph{Technique: Middle-out compression}
In~\cite{indyk2017practical}, every $1$-path is pruned (i.e.~replaced by a long edge) except for its top $\Lambda$ nodes, where $\Lambda$ is an integer parameter.
Combining this ``bottom-out'' compression with the ``top-out'' compression which was introduced in Section~\ref{sec:sketch}, we obtain~\emph{middle-out compression}: every long $1$-path longer than $2\Lambda$ is replaced by a long edge, except for its top and bottom $\Lambda$ nodes.
As we will show in the remainder of this section, applying this pruning rule to the quadtree of~\cite{indyk2017practical} (instead of their ``bottom-out'' rule) is sufficient to obtain a sketch that provably supports out-of-sample approximate nearest neighbor queries. Thus, the sketching algorithm is nearly unchanged.

We remark that in Section~\ref{sec:sketch} we introduced two additional modifications: \emph{grid-net quantization} and~\emph{surrogate hashing}. These were required in order to prove Theorems~\ref{thm:ann_ub} and~\ref{thm:distances_ub}, but in the framework of~\cite{indyk2017practical} they turn out to be unnecessary: grid-net quantization is already organically built into the quadtree approach of~\cite{indyk2017practical}, and surrogate hashing only served to avoid a $O(\log\log n)$ factor in the sketch size (see footnote 4), but in~\cite{indyk2017practical} this factor is tolerated anyway.

\subsection{Sketching Algorithm Recap}
For completeness, let us briefly describe the sketching algorithm of~\cite{indyk2017practical} (the reader is referred to that paper for more formal details), with our modification.
To this end, set
\[ \Lambda = \lceil\log\left(\frac{16d^{1.5}\log\Phi}{\epsilon\delta}\right)\rceil . \]
Suppose w.l.o.g.~that $\Phi$ is a power of $2$. The sketching algorithm proceeds in three steps:
\begin{enumerate}
  \item\emph{Random shifted grids:} Impose a randomly shifted enclosing hypercube on the data points $X$. More precisely, choose a uniformly random shift $\sigma\in\{-\Phi,\ldots,\Phi\}^d$, and set the enclosing hypercube to be $H=[-2\Phi+\sigma_1,2\Phi+\sigma_1]\times[-2\Phi+\sigma_2,2\Phi+\sigma_2]\times\ldots\times[-2\Phi+\sigma_d,2\Phi+\sigma_d]$. Since $X\subset\{-\Phi,\ldots,\Phi\}^d$, it is indeed enclosed by $H$. We then half $H$ along every dimension to create a finer grid with $2^d$ cells, and proceed so (recursively halving every cell along every dimension) to create a hierarchy of nested grids, with $\log(4\Phi)+\Lambda$ hierarchy levels. The top level is numbered $\Phi+2$, which is the log the side length of $H$, and the next levels are decrementing, so  that the grid cells in level $\ell$ have side length $2^\ell$.
  \item\emph{Quadtree construction:} Construct the quadtree which is naturally associated with the nested grids: the root corresponds to $H$, its children correspond to the non-empty cells of the next grid in the hierarchy (a cell is non-empty if it contains a point in $X$), and so on.
Each tree edge is annotated by a bitstring of length $d$, that marks whether the child cell coincides with the bottom half (bit $0$) or the top half (bit $1$) of the parent cell in each dimension.
  \item\emph{Middle-out compression:} For every path of degree-$1$ tree nodes whose length is more than $2\Lambda$, we keep its top $\Lambda$ and bottom $\Lambda$, and replace its remaining middle portion by a long edge. This removes the edge annotations of the middle section (this achieving compression). We label each long edge with the length of the path it replaces.
\end{enumerate}

In the remainder of this section we prove the following.
\begin{theorem}\label{thm:ann_practical}
The above algorithm, with the above setting of $\Lambda$, runs in time $\tilde O(nd(\log\Phi+\Lambda))$ and produces a sketch for the all-nearest-neighbors problem, whose size in bits is
\[
  O\left( n\left(\frac{\log n\cdot(\log\log n + \log(1/\epsilon))}{\epsilon^2} + \log\log\Phi + \log\left(\frac{q}{\delta}\right)\right) + d\log\Phi  + \log\left(\frac{q}{\delta}\right)\log\left(\frac{\log(q/\delta)}{\epsilon}\right) \right).
\]
\end{theorem}

The sketch size is the same as in~\cite{indyk2017practical}, except that we keep at most $2\Lambda$ instead of $\Lambda$ nodes per $1$-path, which increases the sketch size by only a factor of $2$.

\subsection{Basic lemmas}
We start with some useful properties of the above sketch, which are analogous to lemmas from in Section~\ref{sec:sketch}. In the notation below, for a node $v$ in the quadtree, $C(v)$ denotes the subset of points in $X$ that are contained in the grid cell associated with $v$. As in Section~\ref{sec:sketch}, the quadtree is partitioned into a set $\mathcal F(T)$ of~\emph{subtrees} by removing the long edges.

\begin{lemma}[analog of Lemma~\ref{lmm:separation}]\label{lmm:separation_quadtree}
For every point $x\in X$, with probability $1-\delta$, the following holds.
If $z\in\R^d$ is any point outside the grid cell that contains $x$ in level $\ell$ of the quadtree, then $\norm{x-x'}\geq 8\epsilon^{-1}\cdot2^{\ell-\Lambda}\sqrt{d}$.
\end{lemma}
\begin{proof}
The setting of $\Lambda$ is such that with probability $1-\delta$, in every level $\ell$ of the quadtree, the grid cell that contains $x$ also contains the ball at radius $8\epsilon^{-1}\cdot2^{\ell-\Lambda}\sqrt{d}$ around $x$. (This property is known as ``padding''.) The lemma is just a restatement of this property.
See Lemma 1 and Equation (1) in~\cite{indyk2017practical} for details.
\end{proof}

\begin{lemma}[analog of Lemmas~\ref{lmm:subtree_root}, \ref{lmm:subtree_leaf}, \ref{lmm:surrogates}]\label{lmm:samecell}
Let $v$ be a node in the quadtree, and $x,x'\in\R^d$ points contained in the grid cell associated with $v$. Then $\norm{x-x'}\leq2^{\ell(v)}\sqrt{d}$.
\end{lemma}
\begin{proof}
The grid cell associated with $v$ is a hypercube with side $2^{\ell(v)}$ and diameter $2^{\ell(v)}\sqrt{d}$.
\end{proof}

Before proceeding let us make the following point about the quadtree.
\begin{claim}\label{clm:quadtree_levels}
For every leaf $v$ of the quadtree, $C(v)$ contains a single point of $X$, and $v$ is the bottom of a $1$-path of length at least $\Lambda$.
\end{claim}
\begin{proof}
Refining the quadtree grid hierarchy for $\log(4\Phi)$ levels ensures that each grid cell contains at most one point from $X$, and refining for $\Lambda$ additional levels ensures that each leaf is the bottom of a $1$-path of length at least $\Lambda$.
\end{proof}

\begin{claim}\label{clm:leaves}
Every subtree leaf in the quadtree is the bottom of a $1$-path of length at least $\Lambda$.
\end{claim}
\begin{proof}
If $v$ is a leaf of the quadtree, this follows from Claim~\ref{clm:quadtree_levels}. Otherwise this follows from middle-out compression.
\end{proof}

Next we define centers and surrogates.
Centers $c(v)$ are chosen similarly to Section~\ref{sec:sketch}.
The surrogate $s^*(v)$ of every tree node $v$ is simply defined to be the ``bottom-left'' (i.e.~minimal in all dimensions) corner of the grid cell associated with $v$.

\subsection{Approximate Nearest Neighbor Search}
Finally, we can describe the query algorithm and complete its analysis.
Let $y$ be a query point for which we need to report an approximate nearest neighbor from the sketch. The query algorithm is the same as in Section~\ref{sec:ann}: starting with the subtree that contains the quadtree root, it recovers the surrogates in the current subtree and chooses the subtree $v$ whose surrogate is the closest to $y$. If $v$ is a quadtree leaf, its center is returned as the approximate nearest neighbor. Otherwise, the algorithm proceeds by recursion on the subtree under $v$.

\paragraph{Surrogate recovery}
The difference is in the way we recover the surrogates of a given subtree. In Section~\ref{sec:ann} this was done using the surrogate hashes. Here we will use a simpler, deterministic surrogate recovery subroutine.
Let $s^*(H)\in\R^d$ the surrogate of the quadtree root. (We store this point explicitly in the sketch, and it will be convenient to think of it w.l.o.g.~as the the origin in $\R^d$.) As observed in~\cite{indyk2017practical}, for every tree node $v$, if we concatenate the bits annotating the edges on the path from the root to $v$, we get the binary expansion of the point $s^*(H)+s^*(v)$. Therefore, we can recover $s^*(v)$ from the sketch, as long as the path from the root to $v$ does not traverse a long edge.

If the path to $v$ contains long edges (and thus missing bits in the binary expansion of $s^*(v)$), the algorithm completes these bits from the binary expansion of $y$.
Let $r_0,r_1,\ldots$ be the subtree roots traversed by the algorithm, and let $T_0,T_1,\ldots$ be the corresponding subtrees. Let $t$ be the smallest such that the algorithm does not recover the surrogates in $T_t$ correctly (because the bits missing on the long edge connecting $T_{t-1}$ to $T_t$ are not truly equal to those of $y$). As in Section~\ref{sec:ann}, the query algorithm does not know $t$ (it simply always assumes that the bits of $y$ are the correct missing ones), but we will use it for analysis.
Note that by definition of $t$, all surrogates in the subtrees rooted at $r_0,\ldots,r_{t-1}$ are recovered correctly. Thus, the event from Lemma~\ref{lmm:hashes} holds deterministically.

\paragraph{Proof of Theorem~\ref{thm:ann_practical}}
Let $x^*\in X$ be a fixed true nearest neighbor of $y$ in $X$ (chosen arbitrarily if there is more than one). We shall assume that the event in Lemma~\ref{lmm:separation_quadtree} occurs for $x^*$.

\begin{lemma}[analog of Lemma~\ref{lmm:annrounds}]\label{lmm:annrounds_quadtree}
Let $T'\in\mathcal F(T)$ be a subtree rooted in $r$, such that $x^*\in C(r)$.
Let $v$ a leaf of $T'$ that minimizes $\norm{y-s^*(v)}$.
Then either $x^*\in C(v)$,
or every $z\in C(v)$ is a $(1+O(\epsilon))$-approximate nearest neighbor of $y$.
\end{lemma}
\begin{proof}
If $x^*\in C(v)$ then we are done. Assume now that $x^*\in C(u)$ for a leaf $u\neq v$ of $T'$.
Let $\ell:=\max\{\ell(v),\ell(u)\}$. We start by showing that $\norm{y-x^*}>\epsilon^{-1}2^\ell\sqrt{d}$.
Assume by contradiction this is not the case.
Since $x^*\in C(u)$ we have $\norm{x^*-x_{c(u)}}\leq2^{\ell}\sqrt{d}$ by Lemma~\ref{lmm:samecell}, and similarly $\norm{x_{c(u)}-s^*(u)}\leq2^{\ell}\sqrt{d}$. Together, $\norm{y-s^*(u)}\leq(\epsilon^{-1}+2)2^\ell\sqrt{d}$.
On the other hand, by the triangle inequality,
$\norm{y-s^*(v)} \geq \norm{x^*-x_{c(v)}} - \norm{y-x^*} - \norm{x_{c(v)}-s^*(v)}$.
By Claim~\ref{clm:leaves}, both $v$ and $u$ are the bottom of $1$-paths of length at least $\Lambda$, This means that $x^*$ and $x_{c(v)}$ are separated already at level $\ell+\Lambda$, and by Lemma~\ref{lmm:separation_quadtree} this implies $\norm{x^*-x_{c(v)}}\geq8\epsilon^{-1}\cdot2^{\ell}\sqrt{d}$. By the contradiction hypothesis we have $\norm{y-x^*}\leq\epsilon^{-1}2^\ell\sqrt{d}$, and by Lemma~\ref{lmm:samecell}, $\norm{x_{c(v)}-s^*(v)}\leq2^\ell\sqrt{d}$. 
Putting these together yields $\norm{y-s^*(v)}\geq8\epsilon^{-1}\cdot2^{\ell}\sqrt{d}-\epsilon^{-1}2^\ell\sqrt{d}-2^\ell\sqrt{d} > (\epsilon^{-1}+2)2^\ell\sqrt{d}\geq \norm{x_{c(u)}-s^*(u)}$. This contradicts the choice of $v$.

The lemma now follows because for every $z\in C(v)$,
\begin{align}
\norm{y-z} &\leq \norm{y-s^*(v)} + \norm{s^*(v)-x_{c(v)}} + \norm{x_{c(v)}-z} \label{ineq1b} \\
&\leq \norm{y-s^*(u)} + \norm{s^*(v)-x_{c(v)}} + \norm{x_{c(v)}-z} \label{ineq2b} \\
&\leq \norm{y-x^*} + \norm{x^*-x_{c(u)}} + \norm{x_{c(u)}-s^*(u)} +\norm{s^*(v)-x_{c(v)}} + \norm{x_{c(v)}-z} \label{ineq3b} \\
&\leq \norm{y-x^*} + 4\cdot2^\ell\sqrt{d} \label{ineq4b} \\
&\leq (1+4\epsilon)\norm{y-x^*}, \label{ineq5b}
\end{align}
where~(\ref{ineq1b}) and (\ref{ineq3b}) are by the triangle inequality, (\ref{ineq2b}) is since $\norm{y-s^*(v)}\leq\norm{y-s^*(u)}$ by choice of $v$, (\ref{ineq4b}) is by applying Lemma~\ref{lmm:samecell} to each of the last four summands, and~(\ref{ineq5b}) is since we have shown that $\norm{y-x^*}>\epsilon^{-1}2^\ell\sqrt{d}$.
Therefore $z$ is a $(1+4\epsilon)$-approximate nearest neighbor of $y$.
\end{proof}

Now we prove that the query algorithm returns an approximate nearest neighbor for $y$.
We may assume w.l.o.g.~that $\epsilon$ is smaller than a sufficiently small constant.
We consider two cases. In the first case, $x^*\notin C(r_t)$.
Let $i\in\{1,\ldots,t\}$ be the smallest such that $x^*\notin C(r_i)$.
By applying Lemma~\ref{lmm:annrounds_quadtree} on $r_{i-1}$, we have that every point in $C(r_i)$ is a $(1+O(\epsilon))$-approximate nearest neighbor of $y$. After reaching $r_i$, the algorithm would return the center of some leaf reachable from $r_i$, and it would be a correct output.

In the second case, $x^*\in C(r_t)$ contains a true nearest neighbor of $y$. We will show that every point in $C(r_t)$ is a $(1+O(\epsilon))$-approximate nearest neighbor of $y$, so once again, once the algorithm arrives at $r_t$ it can return anything.
By definition of $t$, we know that $y$ does not reside in the grid cell associated with $r_t$. Since $y$ does reside in that cell, we have $\norm{y-x^*}\geq8\epsilon^{-1}2^{\ell(r_t)-\Lambda}\sqrt{d}$ by Lemma~\ref{lmm:separation_quadtree}. On the other hand, by Claim~\ref{clm:leaves}, $r_t$ is the bottom of a $1$-path of length at least $\Lambda$, and therefore any two points in $C(r_t)$ are contained in the same grid cell at level $\ell(r_t)-\Lambda$, whose diameter is $2^{\ell(r_t)-\Lambda}\sqrt{d}$. In particular, for every $x\in C(r_t)$ we have $\norm{x-x^*}\leq2^{\ell(r_t)-\Lambda}\sqrt{d}\leq\frac{1}{8}\epsilon\norm{y-x^*}$. Altogether we get $\norm{y-x} \leq \norm{y-x^*} + \norm{x^*-x} \leq (1+\frac{1}{8}\epsilon)\norm{y-x^*}$, so every $x\in C(r_t)$ is a $(1+\epsilon)$-approximate nearest neighbor of $y$ in $X$.

The proof assumes the event in Lemma~\ref{lmm:separation_quadtree} holds for $x^*$, which happens with probability $1-\delta$. To handle $q$ queries, we can scale $\delta$ down to $\delta/q$ and take a union bound over the $q$ nearest neighbors of the $q$ query points. \qed

\end{document}